\documentclass[12pt]{llncs}

\usepackage{amsmath,amsfonts,amssymb,times,latexsym}

\usepackage[margin=1in]{geometry}

\usepackage{tikz}
\usetikzlibrary{shadows,patterns,shapes,arrows}

\def\stackborderwidth{0.5mm}
\def\stackouterheight{21mm} 
\def\stackouterwidth{7mm} 
\def\stackspacing{10mm} 

\def\setstackspacing{
	\edef\elementspacingx{0}
	\edef\elementspacingy{5mm} 
	\edef\elementoffsetx{3.5mm} 
	\edef\elementoffsety{3mm} 
}

\def\setinputqueuespacing{
	\edef\elementspacingx{5mm} 
	\edef\elementspacingy{0} 
	\edef\elementoffsetx{3mm} 
	\edef\elementoffsety{3mm} 
}

\def\setoutputqueuespacing{
	\edef\elementspacingx{-5mm} 
	\edef\elementspacingy{0} 
	\edef\elementoffsetx{18mm} 
	\edef\elementoffsety{3mm} 
}

\def\initsystem{
	\expandafter\gdef\csname stack 0\endcsname{.}
	\expandafter\gdef\csname stack 1\endcsname{.}
	\expandafter\gdef\csname stack 2\endcsname{.}
	\expandafter\gdef\csname stack 3\endcsname{.}
}

\def\head#1,#2.{#1}
\def\tail#1,#2.{#2}

\def\push#1#2{
		\edef\lista{\csname stack #1\endcsname}
		\expandafter\xdef\csname stack #1\endcsname{#2,\lista}
}

\def\pop#1{
		\edef\lista{\csname stack #1\endcsname}
		\xdef\popped{\expandafter\head\lista}
		\expandafter\xdef\csname stack #1\endcsname{\expandafter\tail\lista.}
}

\def\move#1#2{
    \pop{#1}
    \expandafter\push{#2}{\popped}
}

\def\drawarrow#1#2#3{
		\path (topStack#1) edge [->,bend left] node[auto] {#3} (topStack#2);
}

\tikzset{
    disc/.style={rectangle,
    		minimum height=.5cm,minimum width=.5cm,
    		shade, shading=radial, inner color=#1!20, outer color=#1!60!gray},
    disc 1/.style={disc=yellow},
    disc 2/.style={disc=orange},
    disc 3/.style={disc=red},
    disc 4/.style={disc=green},
    disc 5/.style={disc=blue},
    disc 6/.style={disc=purple},
}

\newcommand{\stack}{
  \fill[black] 
  	(0mm,0mm) rectangle (\stackouterwidth,\stackborderwidth)
    (0mm,\stackborderwidth) rectangle (\stackborderwidth,\stackouterheight)
    (\stackouterwidth-\stackborderwidth,\stackborderwidth) rectangle (\stackouterwidth, \stackouterheight);
}

\newcommand{\queue}{
	\fill[black] 
  	(0mm, 0mm) rectangle (\stackouterheight,\stackborderwidth)
  	(0mm, \stackborderwidth) rectangle (\stackborderwidth,0.5*\stackouterwidth)
  	(\stackouterheight-\stackborderwidth, \stackborderwidth) rectangle (\stackouterheight,0.5*\stackouterwidth);
}

\newcount\curlevel

\def\drawelems#1.{
    \expandafter\if#1..\else
        \rdrawelems#1.
        \advance\curlevel by 1\relax
    \fi
}

\def\rdrawelems#1,#2.{
    \drawelems#2.
    {\edef\n{\the\curlevel}
        \node[disc #1,xshift={\n*\elementspacingx+\elementoffsetx}, yshift={\n*\elementspacingy+\elementoffsety}] {#1};
    }
}

\def\elems#1{
    \curlevel=0
    \expandafter\drawelems#1
}

\def\stackscontent{
	\setoutputqueuespacing
	\begin{scope}[yshift=\stackouterheight-\stackouterwidth]
		\node at (0.5*\stackouterheight,\stackouterwidth) (topStack0) {};
		\queue
		\expandafter\elems\csname stack 0\endcsname
	\end{scope}
	
	\setstackspacing
	\begin{scope}[xshift=\stackouterheight-\stackouterwidth]
	  \foreach \n/\x in {1,2} {
	      \begin{scope}[xshift=\n*\stackspacing]
	      		\node at (0.5*\stackouterwidth,\stackouterheight) (topStack\n) {};
	          \stack
	          \expandafter\elems\csname stack \n\endcsname
	      \end{scope}
	  }
	\end{scope}

	\begin{scope}[xshift=\stackouterheight-\stackouterwidth+3*\stackspacing,yshift=\stackouterheight-\stackouterwidth]
		\setinputqueuespacing
		\node at (0.5*\stackouterheight,\stackouterwidth) (topStack3) {};
		\queue
		\expandafter\elems\csname stack 3\endcsname
	\end{scope}	
}

\newcommand{\abs}[1]{\left|#1\right|}
\newcommand{\length}[1]{\left|#1\right|}
\newcommand{\typeone}[2]{\mathcal{L}_{\textrm{I}}(#1,#2)}
\newcommand{\typetwo}[2]{\mathcal{L}_{\textrm{II}}(#1,#2)}
\newcommand{\typethree}[2]{\mathcal{L}_{\textrm{III}}(#1,#2)}
\newcommand{\perms}[2]{G(#1,#2)}
\newcommand{\states}[1]{\mathfrak{S}_{#1}}
\newcommand{\moves}[1]{M_{#1}}

\begin{document}
\frontmatter

\title{An Improved Lower Bound for Stack Sorting}

\author{Luke Schaeffer}

\institute{School of Computer Science,
University of Waterloo,
Waterloo, ON  N2L 3G1 Canada \\
\email{l3schaef@cs.uwaterloo.ca}
}

\maketitle

\begin{abstract}
We consider the problem of sorting elements on a series of stacks, introduced by Tarjan and Knuth. We improve the asymptotic lower bound for the number of stacks necessary to sort $n$ elements to $0.561 \log_2 n + O(1)$. This is the first significant improvement since the previous lower bound, $\frac{1}{2} \log_2 n + O(1)$, was established by Knuth in 1972. 
\end{abstract}

\section{Introduction}

The subject of our paper is a mathematical puzzle or game for one player. The game is played with $n$ elements labelled from $1$ to $n$ on a \emph{system} of $k$ stacks (also labelled consecutively) arranged in series, with a queue at either end. The queue at one end is called the \emph{input queue} and the queue at the other end is the \emph{output queue}. The game begins with all stacks and queue empty, except for the input queue, which contains the elements (from front to back) $\sigma(1), \sigma(2), \ldots, \sigma(n)$ for some arbitrary permutation $\sigma \in S_n$. The player has the following legal \emph{moves}:
\begin{enumerate}
\setlength{\itemsep}{0pt}
\setlength{\parskip}{0pt}
\item Dequeue an element from the front of the input queue and push it onto stack 1,
\item Pop an element from stack $i$ and push it onto stack $i+1$, for $1 \leq i < k$,
\item Pop an element from stack $k$ and enqueue it on the output queue.
\end{enumerate}
That is, the player may move an element from one stack or queue to the next. Eventually all the elements will reach the output queue, so there are no more legal moves and the game ends. The player wins if, at the end of the game, the elements in the output queue are in sorted order.

The problem of sorting elements with a single stack was introduced by Knuth in \cite{TAOCPv1}, which includes several exercises about stack sorting. For instance, exercise 2.2.1.5 (p. 239) asks the reader to show that $\pi \in S_n$ \emph{cannot} be sorted with one stack if and only if there exist indices $i < j < k$ such that $\pi(j) < \pi(k) < \pi(i)$. Another exercise shows that a single stack can sort exactly $\frac{1}{n+1} \binom{2n}{n}$ permutations. Together these exercises comprise a founding result in the field of permutation patterns.

Tarjan generalized stack sorting to networks of stacks and queues in \cite{tarjan}, including systems of stacks arranged in series. Tarjan showed that a system of $k$ stacks in series can sort $3 \cdot 2^{k-1}$ elements, and cited \cite{TAOCPv3} for a result that $k$ stacks cannot sort $4^{k}$ elements. Knuth revisited stack sorting in \cite{TAOCPv3} and gave two exercises on series of stacks. The exercises concern the minimum of stacks required to sort any permutation of $n$ elements, which we denote $k_n$. Exercise 5.2.4.20 (p. 170) asks for the rate of growth of $k_n$ as a function of $n$. Knuth lists this exercise as an open problem, and it remains an open problem today. Furthermore, there has been no improvement on the elementary bounds 
\begin{align*}
\frac{1}{2} \log_2 n + O(1) &\leq k_n \leq \log_2 n + O(1)
\end{align*}
given in \cite{TAOCPv3}. See \cite{survey} for an excellent survey on stack sorting.

\begin{table}
	\begin{center}
		\begin{tabular}{|c|c|c|c|c|c|c|c|c|c|c|c|c|c|}
		\hline
		$n$   & 0 & 1 & 2 & 3 & 4 & 5 & 6 & 7 & 8 & 9 & 10 & 11 & 12 \\
		\hline
		$k_n$ & 0 & 0 & 1 & 2 & 2 & 2 & 2 & 3 & 3 & 3 & 3 & 3 & 3 \\ 
		\hline
		\end{tabular}
	\end{center}
	\caption{Values of $k_n$ for small $n$.}
	\label{tab:kntable}
\end{table}

\initsystem
\push{0}{1}
\push{0}{3}
\push{0}{2}
\push{0}{4}
  
\begin{figure}[h]
	\caption{Sorting the permutation $\binom{1 \, 2 \, 3 \, 4}{4 \, 2 \, 3 \, 1}$ with 2 stacks using the sequence $m_{121121232333}$ (broken into 7 steps $m_{12}$, $m_{1}$, $m_{12}$, $m_{123}$, $m_{23}$, $m_{3}$, $m_{3}$).}
\begin{center}
  \begin{tikzpicture}
  	\stackscontent
  \end{tikzpicture}
  \move{0}{2}
  \begin{tikzpicture}
  	\stackscontent
  	\drawarrow{0}{2}{$m_{12}$}
  \end{tikzpicture}

	\move{0}{1}
  \begin{tikzpicture}
  	\stackscontent
  	\drawarrow{0}{1}{$m_{1}$}
  \end{tikzpicture}
  \move{0}{2}
  \begin{tikzpicture}
  	\stackscontent
  	\drawarrow{0}{2}{$m_{12}$}
  \end{tikzpicture}
  
  \move{0}{3}
  \begin{tikzpicture}
  	\stackscontent
  	\drawarrow{0}{3}{$m_{123}$}
  \end{tikzpicture}
  \move{1}{3}
  \begin{tikzpicture}
  	\stackscontent
  	\drawarrow{1}{3}{$m_{23}$}
  \end{tikzpicture}
  
  \move{2}{3}
  \begin{tikzpicture}
  	\stackscontent
  	\drawarrow{2}{3}{$m_{3}$}
  \end{tikzpicture}
  \move{2}{3}
  \begin{tikzpicture}
  	\stackscontent
  	\drawarrow{2}{3}{$m_{3}$}
  \end{tikzpicture}
\end{center}
\end{figure}

The main result in this paper is a better lower bound for $k_n$. We will start by re-proving Knuth's lower bound, $k_n \geq \frac{1}{2} \log_2 n + O(1)$. Then we will improve the lower bound to $0.513 \log_2 n + O(1)$ by counting strings of moves, and explore several variations on this approach. Finally, we combine several techniques and obtain a bound of $k_n \geq 0.561 \log_2 n + O(1)$ with the aid of a computer.

\section{Notation and Definitions}

Let us start by introducing some formal notation to describe the game. Define the current \emph{state} of the game to be a description of the elements in each stack and queue (for instance, a list of elements from bottom to top for each stack and front to back for each queue). Let $\states{k}$ be the set of all possible states for a system of $k$ stacks containing finitely many elements, as well as an additional symbol $\varnothing$ to indicate an illegal state. We say a state is initial if all the elements are in the input queue, and a state is final if all the elements are in the output queue. Let $I(n, k)$ be the initial state such that the elements are in the input queue in sorted order. Similarly, let $F(n, k)$ be the final state with the elements in the output queue in sorted order. Note that we win the game if we finish in state $F(n, k)$. 

A \emph{move} is some action by the player that changes the state. We let $m_i$ denote the move that pushes an element onto stack $i$ for each $1 \leq i \leq k$, and let $m_{k+1}$ be the move that pushes to the output queue, so the complete set of moves is $\moves{k} = \{ m_1, \ldots, m_{k+1} \}$. Let $\moves{k}^*$ denote the free monoid on $\moves{k}$, or (abusing notation slightly) the set of finite strings over $\moves{k}$. We adopt the following common notation for string-related concepts. Given $u, v \in \moves{k}^*$, let $uv$ denote their concatenation, let $\abs{u}$ denote the length of $u$ and let $\abs{u}_{i}$ be the number of occurrences of $m_i$ in $u$. Note that we will sometimes write a long sequence of moves $m_{i_1} m_{i_2} \cdots m_{i_j}$ in the more compact notation $m_{i_1 i_2 \cdots i_j}$. 

The relationship between moves and states is described by a monoid action $\ast \colon \moves{k}^{*} \times \states{k} \rightarrow \states{k}$, where sequences of moves act on states. Suppose $w \in \moves{k}^{*}$ is a sequence of moves and $s \in \states{k}$ is a state. Define $w \ast s$ to be the state obtained by starting with state $s$ and performing moves from the sequence $w$ one at a time. Define $w \ast \varnothing := \varnothing$ for all $w$ and $w \ast s := \varnothing$ if some move in $w$ is illegal when we apply the sequence to $s$.

\subsection{Sortable Permutations}
\label{introducingkn}

Suppose we have some $\pi \in S_n$, let $s_{\pi} \in \states{k}$ be the initial state with input queue (from front to back) $\pi(1), \ldots, \pi(n)$ and let $t_{\pi} \in \states{k}$ be the final state with output queue $\pi(1), \ldots, \pi(n)$. We say $\pi$ is \emph{sortable by $k$ stacks} if there exists a sequence of moves $w \in \moves{k}^*$ such that $w \ast s_{\pi} = F(n, k)$. Similarly, we say $\pi$ is \emph{generated by $k$ stacks} if there exists a sequence of moves $w \in \moves{k}^*$ such that $w \ast I(n, k) = t_{\pi}$. 

\begin{definition}
Let $\perms{n}{k} \subseteq S_n$ be the set of permutations generated by $k$ stacks. 
\end{definition}
Recall that our goal is to study the number of stacks necessary to sort $n$ elements. Clearly we can sort all permutations if and only if we can generate all permutations (that is, if $\perms{n}{k} = S_n$).
\begin{definition}
For each $n \geq 0$, let $k_n$ denote the smallest integer such that $\perms{n}{k_n} = S_n$. 
\end{definition}
In other words, $k_n$ is the number of stacks required to sort $n$ elements.
The sequence $\{ k_i \}_{i=1}^{\infty}$ tells us precisely when $\perms{n}{k} = S_n$ according to the following proposition.
\begin{proposition}
For all $n, k \geq 0$, we can sort any permutation of $n$ elements on $k$ stacks (that is, $\perms{n}{k} = S_n$) if and only if $k \geq k_n$. 
\end{proposition}
\begin{proof}
When $k \leq k_n$, the result holds by the definition of $k_n$. When $k > k_n$, we can sort any permutation by first moving all the elements to stack $k - k_n$ and then using the remaining $k_n$ stacks to actually sort them. 
\end{proof} 

\subsection{Complete Strings}

We will study sortable permutations by studying strings of moves. Naturally, we are interested in the set of strings that generate permutations. This motivates the following definitions.

\begin{definition}
We say a string $w \in \moves{k}^*$ is \emph{complete} if $w \ast s$ is a final state for some initial state $s \in \states{k}$. We say $w$ is \emph{$n$-complete} if $w \ast I(n, k)$ is a final state. Let $\typeone{n}{k} \subseteq \moves{k}^{*}$ be the set of all $n$-complete strings.
\end{definition}

The following lemma characterizes $n$-complete strings. 
\begin{lemma}
\label{typeone}
Let $w \in \moves{k}^*$. Then the following are equivalent:
\begin{enumerate}
\setlength{\itemsep}{0pt}
\setlength{\parskip}{0pt}
\setlength{\parsep}{0pt}
\item $w$ belongs to the set $\typeone{n}{k}$.
\item $w$ can be partitioned into $n$ subsequences, all equal to $m_1 m_2 \cdots m_{k+1}$.
\item $\length{u}_{1} \geq \length{u}_{2} \geq \cdots \geq \length{u}_{k+1} \geq 0$ for all prefixes $u$ of $w$ and $\length{w}_{1} = \cdots = \length{w}_{k+1} = n$. 
\end{enumerate}
\end{lemma}
\begin{proof}
\hfill
\begin{itemize}
\setlength\itemindent{1cm}
\item[$(1) \Rightarrow (2)$] \hfill \\
Given an $n$-complete sequence $w \in \moves{k}^*$, we can apply $w$ to the initial state $I(n, k)$. Each symbol in $w$ manipulates a single element when we apply it to a state. We may partition the symbols in $w$ into $n$ subsequences, where each subsequence consists of all moves that manipulate a chosen element. Each element travels from the input stack to the output stack, so the corresponding subsequence is $m_1 m_2 \cdots m_{k+1}$.

\item[$(2) \Rightarrow (3)$] \hfill \\
Since $w$ is composed of $n$ copies of $m_1 m_2 \cdots m_{k+1}$, we have $\length{w}_{i} = n$ for all $i$.  Similarly, any prefix $u$ of $w$ is composed of $n$ prefixes of $m_1 m_2 \cdots m_{k+1}$, so it is clear that $n \geq \length{u}_{1} \geq \length{u}_{2} \geq \cdots \geq \length{u}_{k+1} \geq 0$.
\item[$(3) \Rightarrow (1)$] \hfill \\
The fact that $\length{u}_{1} \geq \length{u}_{2} \geq \cdots \geq \length{u}_{k+1}$ for all prefixes $u$ of $w$ ensures that, when we apply $w$, we never attempt to pop from an empty stack. The sequence $w$ dequeues $n$ elements from the input queue since $\length{w}_{1} = n$, so $w$ is a legal sequence of moves on $I(n, k)$. Furthermore, applying $w$ to $I(n, k)$ moves all the elements to the last stack, since $\length{w}_{1} = \cdots = \length{w}_{k+1} = n$. Hence, $w \in \typeone{n}{k}$. 
\end{itemize}
This completes the proof.
\end{proof}
We note that any string in $\typeone{n}{k}$ contains $m_i$ exactly $n$ times, and therefore has length $n(k+1)$. Hence, we also define
\begin{align*}
\typetwo{n}{k} &:= \{ u \in \moves{k}^* : \length{u}_{1} = \cdots = \length{u}_{k+1} = n \} \\
\typethree{n}{k} &:= \{ u \in \moves{k}^* : \length{u} = n(k+1) \},
\end{align*}
and note that $\typeone{n}{k} \subseteq \typetwo{n}{k} \subseteq \typethree{n}{k}$ for all $n, k \geq 0$. We call the three languages type I, type II and type III respectively. In many cases, $\typetwo{n}{k}$ and $\typethree{n}{k}$ are easier to work with than $\typeone{n}{k}$, and one can show that the three languages have the same number of strings up to a polynomial factor. 

\section{Known Lower Bound}\label{bounds}

In this section we will prove Knuth's lower bound for $k_n$, but we need a lemma first. Define the product of two sets of permutations, $A, B \in S_n$ as $AB := \{ ab : a \in A, b \in B \}$.

\begin{lemma} \label{oneforall}
For all $n \geq 0$ and $k \geq 0$, we have $\perms{n}{k} = \perms{n}{1}^k = \{ \pi_1 \circ \ldots \circ \pi_{k} \colon \pi_1, \ldots, \pi_k \in \perms{n}{1} \}$.
\end{lemma}
\begin{proof}
Suppose we have a string $w \in \typeone{n}{k}$, and we apply it to an initial state with $n$ elements. Consider the $i$th stack and note that the elements enter in the order $x_1, \ldots, x_n$ and leave in the order $\pi_{i}(x_1), \ldots, \pi_i(x_n)$, for some $\pi_i \in \perms{n}{1}$. Then the permutation generated by the system is clearly the composition $\pi_{k} \circ \cdots \circ \pi_{1} \in \perms{n}{1}^{k}$, so $\perms{n}{k} \subseteq \perms{n}{1}^{k}$. 

In the other direction, let $\pi_1, \ldots, \pi_k$ be permutations in $\perms{n}{1}$. There exist strings $w_1 \in \{ m_1, m_2 \}^{*}$, $w_2 \in \{ m_2, m_3 \}^{*}$, \ldots, $w_k \in \{ m_k, m_{k+1} \}^{*}$ such that stack $i$ generates $\pi_i$ when we apply $w_i$, for each $1 \leq i \leq k$. We would like to find $w \in \typeone{n}{k}$ such that $w_i$ is the subsequence of moves in $w$ that manipulate stack $i$, for each $i$. If we can find such a string, then it must generate $\pi_k \circ \cdots \circ \pi_{1}$, and therefore show that $\perms{n}{1}^{k} \subseteq \perms{n}{k}$. 

Let us start by numbering the symbols in $w_1, \ldots, w_k$ so that the $j$th occurrence of $m_i$ is written $m_{i}^{(j)}$. Define a relation $\prec$ such that $m_{i_1}^{(j_1)} \prec m_{i_2}^{(j_2)}$ if $m_{i_1}^{(j_1)}$ is before $m_{i_2}^{(j_2)}$ in some $w_\ell$, and then close $\prec$ under transitivity to obtain a partial order (we leave it as an exercise to show that this can be done). Now we extend $\prec$ to a total order and list the symbols in sorted order to obtain $w \in \typeone{n}{k}$.
\end{proof}
Now that we have Lemma~\ref{oneforall}, we can prove Knuth's lower bound for $k_n$. This is not a new result, but we will use the same ideas for later theorems, so it is a useful example. 
 
\begin{theorem} \label{knownlowerbound}
For all $n, k \geq 0$ we have $\abs{\perms{n}{k}} \leq 4^{nk}$, from which it follows that $\frac{1}{2} \log_2 n + O(1) \leq k_n$.
\end{theorem}
\begin{proof}
The lemma gives us $\perms{n}{k}=\perms{n}{1}^{k}$ and hence $\abs{\perms{n}{k}}=\abs{\perms{n}{1}^k}\leq\abs{\perms{n}{1}}^{k}$. Every permutation is generated by some string in $\typeone{n}{k}$, so $\abs{\perms{n}{1}} \leq \abs{\typeone{n}{1}} \leq \abs{\typethree{n}{1}} = 4^{n}$. We conclude that $\abs{\perms{n}{k}} \leq 4^{nk}$ for all $n, k \geq 0$. 

Now recall that $\perms{n}{k_n} = S_n$ by definition, so $n! = \abs{\perms{n}{k_n}} \leq 4^{nk_n}$. Taking logarithms gives $\log_2 n! \leq 2nk_n$, and then we apply Stirling's approximation to obtain $\frac{1}{2} \log_2 n + O(1) \leq k_n$. 
\end{proof}

\section{Working Towards a Better Lower Bound}

We claim that the lower bound presented in Theorem~\ref{knownlowerbound} can be improved. The key is to show that $\abs{\perms{n}{k}}$ grows exponentially slower than $4^{nk}$, which is the bound we use in the theorem, derived from the inequality $\abs{\perms{n}{k}} \leq \abs{\perms{n}{1}}^{k}$. Unfortunately, $\abs{\perms{n}{1}}^{k}$ is a poor upper bound for $\abs{\perms{n}{k}}$, even for small $n$ and $k$. For example, $\abs{\perms{4}{2}} = \abs{S_4} = 24$, but $\abs{\perms{4}{1}}^2 = 14^2 = 196$. 

Conceptually, the inequality $\abs{\perms{n}{k}} \leq  \abs{\perms{n}{1}}^{k}$ comes from breaking the system of stacks into $k$ single stack systems, and assuming those systems do not interact. The problem is that adjacent stacks \emph{do} interact, and one stack often undoes work done by another stack. For instance, if the first stack reverses the elements and the second stack reverses them again, then we have accomplished nothing with those stacks. 

Now suppose that we break the system into groups of $\ell$ stacks instead of single stacks, where $\ell > 1$ is a constant. If we can show that $\abs{\perms{n}{\ell}}$ grows slower than $\abs{\perms{n}{1}}^{\ell}$, then we can prove a better lower bound with the following proposition. 
\begin{proposition} \label{propforbound}
Let $\ell \geq 1$ be a constant. If $\abs{\perms{n}{\ell}} \in O(b^{n})$ then $k_n \geq \frac{\ell \log_2 n}{\log_2 b} + O(1)$.
\end{proposition}
\begin{proof}
We can use Lemma~\ref{oneforall} to show that 
\begin{align*}
n! &= \abs{\perms{n}{k_n}} \leq \abs{\perms{n}{\ell}^{\lceil k_n/\ell \rceil}} \leq \abs{\perms{n}{\ell}}^{k_n/\ell+1} \leq cb^{n(k_n/\ell+1)}.
\end{align*}
Taking logarithms and using Stirling's approximation gives
\begin{align*}
n \log_2 n + O(n) &\leq \log_2 \left[ cb^{n(k_n / \ell+1)} \right] \\
\log_2 n + O(1) &\leq (k_n / \ell) \log_2 b \\
\frac{\ell \log_2 n}{\log_2 b} + O(1) &\leq k_n,
\end{align*}
completing the proof. 
\end{proof}
In the case $\ell = 1$ with $\abs{\perms{n}{1}} \in O(4^n)$ we get $\frac{1}{2} \log_2 n + O(1) \leq k_n$, which is the lower bound from Theorem~\ref{knownlowerbound}. In this section we will focus on $\ell = 2$ since it is the easiest case after $\ell = 1$. Unfortunately, we cannot use $\abs{\typeone{n}{2}}$ as a bound for $\abs{\perms{n}{2}}$ (as we did for the one stack case) since $\abs{\typeone{n}{2}} \approx \abs{\typethree{n}{2}} = 27^n > 16^n$. Evidently the ratio of $n$-complete strings to generate permutations, $\frac{\abs{\typeone{n}{2}}}{\abs{\perms{n}{2}}}$, is about $(27/16)^n$. Contrast this with one stack, where it turns out that $\abs{\typeone{n}{1}} = \abs{\perms{n}{2}}$.

In the next section, we define an equivalence relation that expresses when two strings generate the same permutation. We will show that we can rewrite strings in $\abs{\typethree{n}{2}}$ to a (more or less) canonical string in the same equivalence class. Then we will count only the canonical strings to bound $\abs{\perms{n}{2}}$ and obtain a better lower bound for $k_2$. 

\subsection{Equivalence of Strings} \label{equivalenceOfStrings}

Recall that $\moves{k}^{*}$ is a monoid that acts on game states $\states{k}$ via $\ast$. Let us define an equivalence relation $\sim$ on $\moves{k}^*$ such that $x \sim y$ if $x \ast s = y \ast s$ for all $s \in \states{k}$. In other words, two strings $x$ and $y$ are equivalent if they act on every state in exactly the same way.

We note three properties of $\sim$ without proof.
\begin{proposition}
The relation $\sim$ is a congruence relation. That is, if $u \sim v$ and $x \sim y$ then $ux \sim vy$. 
\end{proposition}

\begin{proposition} \label{onestateisenough}
Suppose $u, v \in \moves{k}^*$ are strings. Then $u \ast s = v \ast s \neq \varnothing$ for some $s \in S_k$ if and only if $u \sim v$. 
\end{proposition}

\begin{lemma} \label{cancellative}
For strings $a, b, u, v \in \moves{k}^*$, we have $aub \sim avb$ if and only if $u \sim v$.
\end{lemma}

For example, every string in $\typeone{2}{2}$ generates a permutation in $\perms{2}{2}$. Therefore $\sim$ splits $\typeone{2}{2}$ into two congruence classes: $m_{123123} \sim m_{121323} \sim m_{112233}$ and $m_{112323} \sim m_{121233}$. Using cancellation, we can deduce a set of relations, $R := \{ m_{13} \sim m_{31}, m_{1223} \sim m_{2312}, m_{1232} \sim m_{2123} \}$. 

The first relation, $m_{13} \sim m_{31}$, says that performing move $m_1$ followed by $m_3$ is indistinguishable from performing $m_3$ and then $m_1$. This makes sense, since $m_1$ affects the input queue and stack 1, while $m_3$ affects stack 2 and the output queue. Since $m_1$ and $m_3$ operate on different pieces of the system, they should commute. The other two relations, $m_{12} m_{23} \sim m_{23} m_{12}$ and $m_{123} m_2 \sim m_{2} m_{123}$ are also commutativity relations $uv \sim vu$ for strings $u,v \in \moves{2}^{*}$ that operate on different parts of the system.

We can generate a congruence relation $\sim_R$ (based on the relations in $R$) such that $m_{13} \sim_R m_{31}, m_{1223} \sim_R m_{2312}, m_{1232} \sim_R m_{2123}$. It follows that for $u, v \in \moves{2}$, if $u \sim_R v$ then $u \sim v$. In other words, $\sim_{R}$ is finer relation than $\sim$. One can show that $m_{122233} \sim m_{223312}$ but $m_{122233} \not \sim_R m_{223312}$, so $\sim_{R}$ is strictly finer than $\sim$.  

\subsection{String Rewriting}

Recall that every string in $\typeone{n}{2}$ generates a permutation. Since two strings $u,v \in \moves{2}^*$ generate the same permutation if and only if $u \sim v$, the congruence classes of $\typeone{n}{2}$ correspond exactly to permutations in $\perms{n}{2}$. We would like to find a set of representative strings, one for each equivalence class of $\moves{2}^{*}$. Unfortunately, we do not know how to do this for $\sim$. Instead we will find a set of representatives $\mathcal{V}$ for the equivalence classes of $\moves{2}^{*}$ under the relation $\sim_{R}$. 

Let $\mathcal{V}$ contain the lexicographically maximal string from each equivalence class of $\moves{2}^{*}$. Suppose $w$ is a string in $\mathcal{V}$ containing $m_{13}$ as a substring. Then we can write $w = a m_{13} b$ for some $a, b \in \moves{2}^{*}$ and note that $w' := a m_{31} b$ is lexicographically larger than $w$. But $w'$ is in the same equivalence class as $w$ since $w \sim_{R} w'$, so we have a contradiction. Therefore all strings in $\mathcal{V}$ avoid having $m_{13}$ as a substring. There is a similar argument for the other relations in $R$, so strings in $\mathcal{V}$ avoid $m_{1223}$ and $m_{1232}$ as well. 

Define $\mathcal{U} \subseteq \moves{k}^{*}$ to be the set of strings that do not contain $m_{13}$, $m_{1223}$ or $m_{1232}$ as substrings. We have just shown that $\mathcal{V} \subseteq \mathcal{U}$. It turns out that $\mathcal{U} = \mathcal{V}$, but this is more than we need, so we omit the proof\footnote{The proof amounts to showing that the rewriting rules $m_{13} \rightarrow m_{31}$, $m_{1223} \rightarrow m_{2312}$, $m_{1232} \rightarrow m_{2123}$ form a confluent system.}.

\subsection{Enumerating Strings in $\mathcal{U}$}

Let $\pi$ be an arbitrary permutation in $\perms{n}{2}$, generated by some string $w$ in $\typeone{n}{2}$. This string has a representative $w'$ in $\mathcal{V} \subseteq \mathcal{U}$ such that $w \sim_{R} w'$ and hence $w'$ also generates $\perms{n}{2}$. It follows that $w'$ is also in $\typeone{n}{2} \subseteq \typethree{n}{2}$. Thus, any permutation in $\perms{n}{2}$ is generated by a string in the set $\typethree{n}{2} \cap \mathcal{U}$. Each string in $\typethree{n}{2} \cap \mathcal{U}$ can generate at most one permutation, so we have $\abs{\perms{n}{2}} \leq \abs{\typethree{n}{2} \cap \mathcal{U}}$. Note that $\typethree{n}{2}$ is the set of strings in $\moves{2}^{*}$ of length $3n$, so if we can find a generating function $U(x)$ that counts strings in $\mathcal{U}$ (weighted by the length of the string), then 
\begin{align*}
\abs{\perms{n}{2}} \leq \abs{\typethree{n}{2} \cap \mathcal{U}} = [x^{3n}] U(x). 
\end{align*}

There are several published techniques for counting strings avoiding some finite set of forbidden substrings. For example, there is a technique using systems of equations by Guibas and Odlyzko \cite{guibasodlyzko}, or the Goulden-Jackson cluster method in~\cite{clustermethod,combenum}. Using either method, the generating function for $U(x)$ is 
\begin{align*}
U(x) &= \frac{1}{1 - 3x + x^2 + 2x^4}.
\end{align*}
The denominator has distinct roots, so the partial fraction decomposition is
\begin{align*}
U(x) &= \sum_{i=1}^{4} \frac{c_i}{\lambda_i - x}
\end{align*}
where $c_1, \ldots, c_4, \lambda_1, \ldots, \lambda_4 \in \mathbb C$, with $\lambda_1, \ldots, \lambda_4$ are the roots of $1 - 3x + x^2 + 2x^4$:
\begin{align*}
\lambda_1 &\doteq 0.40671 & \lambda_3 &\doteq -0.59149+1.1108i \\
\lambda_2 &\doteq 0.77626 & \lambda_4 &\doteq -0.59149-1.1108i.
\end{align*}
Each of the four terms is a geometric series and it is not difficult to show that $[x^n] \frac{c_i}{\lambda_i - x}$ is in $\Theta(\abs{\lambda_i}^{-n})$ for all $i$. The first term dominates, so $[x^n] U(x)$ is in $\Theta(\abs{\lambda_1}^{-n}) \doteq \Theta(2.45875^n)$. It follows that 
\begin{align*}
\abs{\perms{n}{2}} &\leq \abs{\typethree{n}{2} \cap \mathcal{U}} \\
&= [x^{3n}] U(x) \\
&\in O(14.864^n).
\end{align*}
We can now prove our first improvement to the lower bound: 
\begin{theorem} \label{mainresult}
For all $n$, 
\begin{align*}
k_n \geq \frac{2 \log_2 n}{\log_2 14.864} + O(1) \doteq 0.51364 \log_2 n + O(1).
\end{align*}
\end{theorem}
\begin{proof}
Apply Proposition~\ref{propforbound} to the bound $\abs{\perms{n}{2}} \leq [x^{3n}] U(x) \in O(14.864^n)$. 
\end{proof}

\section{Further Improvements}
\label{furtherimprovements}

In this section we discuss techniques to further improve our lower bound in the previous section. These techniques give us the lower bounds $k_n \geq 0.52224 \log_2 n + O(1)$ and $k_n \geq 0.53028 \log_2 n + O(1)$ respectively. Then we combine these techniques with a large set of relations, $R_{16}$, to obtain our final bound,
\begin{align*}
k_n \geq 0.561 \log_2 n + O(1).
\end{align*}

\subsection{Additional Relations}

For our original lower bound, we used the three relations $m_{13} \sim_R m_{31}$, $m_{1223} \sim_R m_{2312}$ and $m_{1232} \sim_R m_{2123}$. Recall that the congruence relation $\sim_R$ is strictly finer than $\sim$, so there exist elements $u, v$ that are equivalent under $\sim$ but not under $\sim_{R}$. By adding more relations to $R$, we reduce the number of equivalence classes, and therefore improve the bound on $k_n$. 

Suppose we let $R' = R \cup \{ m_{112223} \sim m_{231122}, m_{122233} \sim m_{223312} \}$ be a new set of relations. By the same process as before, we obtain a set of strings $\mathcal{U}' \subseteq \moves{2}^*$ avoiding $m_{13}$, $m_{1223}$, $m_{1232}$, $m_{112223}$ and $m_{122233}$. The corresponding generating function is 
\begin{align*}
\hat{U}(x) &= \frac{1}{1 - 3x + x^2 + 2x^4 + 2x^6}
\end{align*}
The smallest real root of the denominator is approx. $0.41278$, leading to $[x^n] \hat{U}(x) \in \Theta(2.4229^n)$ and $\abs{\perms{n}{2}} \in O(14.218^n)$. In the end, we get the lower bound $k_n \geq 0.52224 \log_2 n + O(1)$.

We see that additional relations improve our lower bound. For our final result we will use $R_{16}$, the set of all nontrivial relations $u \sim v$ such that $\abs{u}, \abs{v} \leq 16$. We use a computer to generate the 1591 relations in $R_{16}$\footnote{See \texttt{http://www.student.cs.uwaterloo.ca/\textrm{\symbol{126}}l3schaef/stacksort/} for the list of relations.}. The algorithm for finding relations is roughly as follows. 
\begin{enumerate}
\item We consider all strings of length $\ell$, and let $s$ be the state with $\ell$ elements in the input queue and $\ell$ elements in each queue. 
\item Then we compute $w \ast s$ for each string $w \in \moves{2}^{*}$ of length $\ell$, and note that $w \ast s \neq \varnothing$ since we do not have enough moves to empty a stack and pop from it.
\item Finally, the equivalence classes of $w \ast s$ (under $=$) correspond to the equivalence classes of $w$ (under $\sim$) by Proposition~\ref{onestateisenough}.
\end{enumerate}
It is clear that the algorithm is computationally expensive, and although there are ways to improve it, the running time is unavoiably exponential in $\ell$. We draw the line at $R_{16}$ because of the cost of computing new relations and solving the resulting equations, as well as the rapidly diminishing contribution (in terms of the lower bound) of each new relation.

There are some patterns among the relations we computed, although most of the relations do not fit any known pattern. Consider the following relations in $R'$
\begin{align*}
m_{12} m_{23} &\sim m_{23} m_{12} \\
m_{1122} m_{23} &\sim m_{23} m_{1122} \\
m_{12} m_{2233} &\sim m_{2233} m_{12}.
\end{align*}
These relations are the first three examples of an infinite set of relations of the form $(m_{1} u m _{2}) (m_{2} v m_{3}) \sim (m_{2} v m_{3}) (m_{1} u m_{2})$ for $u, v \in \moves{2}^{*}$ \footnote{There are additional conditions on $u$ and $v$. Specifically, $u$ must transfer elements from the input queue to stack 2, and $v$ must transfer elements from stack 1 to the output queue}. Similarly, there are relations of the form $u m_{2} \sim m_{2} u$ for $u \in \moves{2}^{*}$ a complete string, although not every complete string gives a useful relation. For instance, $u = m_{123} m_{123}$ is complete, but since we can already deduce $u m_{2} \sim m_{2} u$ from the relation $m_{123} m_{2} \sim m_{2} m_{123}$, the relation $u m_{2} \sim m_{2} u$ is redundant. 

\subsection{Weighted Generating Functions}

Recall that to prove our main result, we found a generating function $U(x)$ for the set $\mathcal{U}$, weighted by string length. If we use a different weight function, which sends $w$ to $x_1^{\abs{w}_{1}} x_2^{\abs{w}_{2}} x_3^{\abs{w}_{3}}$ then we get a multivariate generating function $U(x_1, x_2, x_3)$ for $\mathcal{U}$. Then the Guibas-Odlyzko method or the Goulden-Jackson cluster method can be used to compute $U(x_1, x_2, x_3)$:
\begin{align*}
U(x_1, x_2, x_3) &= \frac{1}{1 - x_1 - x_2 - x_3 + x_1 x_3 + 2x_1 x_2^2 x_3}.
\end{align*}
We have seen that $\abs{\perms{n}{2}} \leq \abs{\typeone{n}{2} \cap \mathcal{U}} \leq \abs{\typetwo{n}{2} \cap \mathcal{U}}$ and hence
\begin{align*}
\abs{\perms{n}{2}} &\leq \abs{\typeone{n}{2} \cap \mathcal{U}} \leq \abs{\typetwo{n}{2} \cap \mathcal{U}} = [x_1^n x_2^n x_3^n] U(x_1, x_2, x_3).
\end{align*}
Notice that since all the coefficients of $U(x_1, x_2, x_3)$ are positive integers, if we have a ring homomorphism $\varphi : \mathbb{Z}[[x_1, x_2, x_3]] \rightarrow \mathbb{Z}[[x]]$ then 
\begin{align*}
[x_1^n x_2^n x_3^n] U(x_1, x_2, x_3) &\leq [\varphi(x_1)^n \varphi(x_2)^n \varphi(x_3)^n] \varphi(U(x_1, x_2, x_3)).
\end{align*}
This allows us to give $m_{1}$, $m_{2}$ and $m_{3}$ different weights, as in the following proposition. 
\begin{proposition}
For all $n \geq 0$, 
\begin{align*}
\abs{\perms{n}{2}} &\leq [x^{n(\alpha_1 + \alpha_2 + \alpha_3)}] \left( \frac{1}{1 - x^{\alpha_1} - x^{\alpha_2} - x^{\alpha_3} + x^{\alpha_1 + \alpha_3} + 2x^{\alpha_1 + 2\alpha_2 + \alpha_3}} \right).
\end{align*}
\end{proposition}
\begin{proof}
Take $\varphi$ to be the ring homomorphism such that $x_i \mapsto x^{\alpha_i}$ for $i=1,2,3$. The result follows easily. 
\end{proof}

For example, take $\alpha_1 = \alpha_3 = 1$ and $\alpha_2 = 2$. We get 
\begin{align*}
\abs{\perms{n}{2}} &\leq [x^{4n}]  \left( \frac{1}{1 - 2x + 2x^{6}} \right) \\
&\in O(13.708^n).
\end{align*}
Applying Proposition~\ref{propforbound} gives
\begin{align*}
k_n &\geq 0.52953 \log_2 n + O(1).
\end{align*}
This is quite an improvement over our first result, considering that only the analysis changed. Unfortunately, it is not immediately obvious how to choose values for $\alpha_1, \alpha_2, \alpha_3 \in \mathbb{N}$ to get the best possible bound. 

The asymptotic behaviour of a univariate generating function is connected to its radius of convergence, and therefore to the location of its poles. Let us assume that multivariate generating functions are similar, and look at the zeros of the denominator. In particular, let us try to find $x_1, x_2, x_3 \in (0, \infty)$ positive real numbers such that our polynomial denominator $1-x_1-x_2-x_3+x_1x_3+2x_1x_2^2x_3$ is zero and $x_1 x_2 x_3$ is minimized (or equivalently, $\frac{1}{x_1 x_2 x_3}$ is maximized). Using Lagrange multipliers, the solution satisfies the following equations:
\begin{align*}
0 &= 1-x_1-x_2-x_3+x_1 x_3+2x_1 x_2^2 x_3 \\
0 &= x_1 x_3 + \lambda(4x_1 x_2 x_3-1) \\
0 &= x_2 x_3 + \lambda(x_3+2x_2^2 x_3-1) \\
0 &= x_1 x_2 + \lambda(x_1+2x_1 x_2^2-1).
\end{align*}
In this case, it is possible to solve these equations exactly by hand, but in general we employ a computer algebra system and obtain numerical results. The optimal solution is $x_1 = x_3 = \frac{1}{2}$ and $x_2 = 1 - \frac{\sqrt{2}}{2}$, giving $\frac{1}{x_1 x_2 x_3} = 8 + 4\sqrt{2} \doteq 13.659$ as the growth rate. Therefore, we expect that $[x_1^n x_2^n x_3^n] U(x_1, x_2, x_3)$ is in $O((8 + 4\sqrt{2} + \varepsilon)^n)$ for all $\varepsilon > 0$, with a corresponding lower bound is $k_n \geq 0.53029 \log_2 n + O(1)$. 

Now we will determine the optimal weights, and see an example to illustrate that close-to-optimal weights give a close-to-optimal lower bound. Our optimal weights should be such that $x_1 = x^{\alpha_1} = x^{\alpha_3} = x_3$ and $x_2 = x^{\alpha_2}$. It follows that the ratio of two weights is
\begin{align*}
\frac{\alpha_2}{\alpha_1} = \frac{x_2}{x_1} = \frac{\log \left(1 - \frac{\sqrt{2}}{2} \right)}{\log(1/2)} \doteq 1.77155.
\end{align*}
Hence, we take integer weights $\alpha_1 = \alpha_3 = 4$ and $\alpha_2 = 7$ which give the ratio $\alpha_2 / \alpha_1 = 1.75 \approx 1.77155$. With these weights we obtain 
\begin{align*}
[x^{(\alpha_1 + \alpha_2 + \alpha_3)n}] U(x^{\alpha_1}, x^{\alpha_2}, x^{\alpha_3}) &= [x^{15n}] U(x^4, x^7, x^4) \in O(13.657^n)
\end{align*}
and a corresponding lower bound of $k_n \geq 0.53028 \log_2 n + O(1)$, which is virtually indistinguishable from the bound we would expect with optimal weights. 

\subsection{Main Lower Bound}

For our main result, we apply the weighted generating function technique to a large set of relations, $R_{16}$. For simplicity, we assume the weights for $m_1$ and $m_3$ are equal, so $x_1 = x_3$. We obtain a generating function $U(x_1, x_2, x_3) = \frac{1}{p(x_1, x_2, x_3)}$ where
\begin{align*}
p(x_1, x_2, x_1) :=& 1 - 2x_1 - x_2 + 
x_1^{2} +
2 x_1^{2} x_2^{2} +
2 x_1^{3} x_2^{3} +
2 x_1^{4} x_2^{3} +
5 x_1^{4} x_2^{4} +
4 x_1^{5} x_2^{4} +
14 x_1^{5} x_2^{5} + \\
&
8 x_1^{6} x_2^{4} +
13 x_1^{6} x_2^{5} +
42 x_1^{6} x_2^{6} +
22 x_1^{7} x_2^{5} +
40 x_1^{7} x_2^{6} +
41 x_1^{8} x_2^{5} +
132 x_1^{7} x_2^{7} +
77 x_1^{8} x_2^{6} + \\
&
123 x_1^{8} x_2^{7} +
134 x_1^{9} x_2^{6} +
429 x_1^{8} x_2^{8} +
252 x_1^{9} x_2^{7} +
248 x_1^{10} x_2^{6}.
\end{align*}
We wish to maximize $\frac{1}{x_1^2 x_2}$ subject to $p(x_1, x_2, x_1) = 0$, and obtain optimal weights from the solution. The optimal solution is at $x_1 \doteq 0.47565, x_2 \doteq 0.37405$ with objective value $\beta := \frac{1}{x_1^2 x_2} \doteq 11.817$. Then by choosing the right weights we get $\abs{\perms{n}{2}} \leq \abs{\typetwo{n}{2} \cap \mathcal{U}} \in O((\beta + \varepsilon)^{n})$ for arbitrarily small $\varepsilon > 0$. For sufficiently small $\varepsilon$, this gives us the bound
\begin{align*}
k_n \geq 0.56136 \log_2 n + O(1).
\end{align*}

\section{Conclusion}

In summary, we gave a proof of Knuth's lower bound ($k_n \geq \frac{1}{2} \log_2 n + O(1)$) using a counting argument that bounded $\perms{n}{k}$. We generalized the counting argument from single stacks to pairs of stacks and discovered that many sequences of moves may generate a permutation. We found an equivalence relation on strings of moves, and used string rewriting techniques to reduce the number of strings, leading to a new lower bound, $k_n \geq 0.513 \log_2 n + O(1)$. 

In Section \ref{furtherimprovements} we discussed techniques for improving the bound using more relations or more sophisticated analysis. The combination of these techniques and a set ($R_{16}$) of 1591 relations give us our main result,
\begin{align*}
k_n &\geq 0.561 \log_2 n + O(1).
\end{align*}
This bound can almost certainly be improved by using more relations, and finding the true asymptotic complexity of $k_n$ remains an open problem. 

\section{Acknowledgments}

I would like to thank Ming Li for bringing this problem to my attention, and Jeffrey Shallit for his assistance preparing this paper. 

\bibliography{stackbib}{}
\bibliographystyle{plain}
\end{document}